\documentclass[journal,draftcls,onecolumn,12pt,twoside]{IEEEtran}

\usepackage[T1]{fontenc} 
\usepackage{times}
\usepackage{amsmath}
\usepackage{amssymb}
\usepackage{amsthm}
\usepackage{color}
\usepackage{algorithm}
\usepackage[noend]{algorithmic}
\usepackage{graphicx}
\usepackage{subfig}
\usepackage{multirow}
\usepackage[bookmarks=false,colorlinks=false,pdfborder={0 0 0}]{hyperref}
\usepackage{cite}
\usepackage{bm}
\usepackage{arydshln}
\usepackage{mathtools}
\usepackage{microtype}

\usepackage[figuresright]{rotating}
\usepackage{threeparttable}
\usepackage{booktabs}

\newtheorem{theorem}{Theorem}

\newtheorem{lemma}[theorem]{Lemma}
\newtheorem{example}[theorem]{Example}

\long\def\symbolfootnote[#1]#2{\begingroup
\def\thefootnote{\fnsymbol{footnote}}\footnote[#1]{#2}\endgroup}
\renewcommand{\paragraph}[1]{{\bf #1}}

\title{Two Piggybacking Codes with Flexible Sub-Packetization to Achieve Lower Repair Bandwidth}

\author{Hao Shi, Zhengyi Jiang, Zhongyi Huang, Bo Bai, Hanxu Hou\\
\thanks{
This paper was presented in part at the IEEE International Symposium on Information Theory (ISIT), 2022 \cite{Shi2206:Piggybacking}.
H. Shi, Z. Jiang and Z. Huang are with the Department of Mathematics Sciences, Tsinghua University (E-mail: shih22@mails.tsinghua.edu.cn, jzy21@mails.tsinghua.edu.cn, zhongyih@tsinghua.edu.cn).
B. Bai and H. Hou are with the Theory Lab, Central Research Institute, 2012 Labs, Huawei Tech. Co. Ltd.~(E-mail: baibo8@huawei.com,
hou.hanxu@huawei.com).
This work was partially supported by the National Key R\&D Program of
China (No. 2020YFA0712300), the National Natural Science Foundation of China (No. 62071121).}}

\begin{document}
\let\emph\textit
\maketitle
\pagestyle{empty}  
\thispagestyle{empty} 
\begin{abstract}
	
As a special class of array codes, $(n,k,m)$ piggybacking codes are
MDS codes (i.e., any $k$ out of $n$ nodes can retrieve all data symbols) that
can achieve low repair bandwidth for single-node failure with low sub-packetization $m$.
In this paper, we propose two new piggybacking codes that have lower repair bandwidth than the existing
piggybacking codes given the same parameters. Our first piggybacking codes can
support flexible sub-packetization $m$ with
$2\leq m\leq n-k$, where $n - k > 3$. We show that our first piggybacking codes have
lower repair bandwidth for any single-node failure than the existing piggybacking
codes when $n - k = 8,9$, $m = 6$ and $30\leq k \leq 100$. Moreover, we
propose second piggybacking codes such that the sub-packetization is a multiple of
the number of parity nodes (i.e., $(n-k)|m$), by jointly designing the piggyback function for
data node repair and transformation function for parity node repair. We show that the proposed
second piggybacking codes have lowest
repair bandwidth for any single-node failure among all the existing piggybacking codes
for the evaluated parameters $k/n = 0.75, 0.8, 0.9$ and $n-k\geq 4$.
\end{abstract}


\section{introduction}

Maximum distance separable (MDS) array codes are widely employed in the modern
distributed storage systems because they provide the maximum data reliability for a
given level of storage overhead. An $(n,k,m)$ MDS array code encodes $km$
{\em data symbols} into $nm$ {\em coded symbols} that are equally stored in $n$ nodes,
where each node stores $m$ symbols. We call the number of symbols stored in each node
as the sub-packetization level. The $(n,k,m)$ MDS array codes satisfy the {\em MDS property}
that is any $k$ out of $n$ nodes can retrieve all $km$ data symbols. The codes are referred
to as {\em systematic codes} if the $km$ data symbols are included in the first $k$ nodes
of obtained $n$ nodes. Reed-Solomon (RS) codes \cite{reed1960} are typical MDS array codes
with $m = 1$. In this paper, we consider systematic
MDS array codes that contain $k$ {\em data nodes} which store the $km$ data symbols
and $r=n-k$ {\em parity nodes} which store the $rm$ parity symbols.

In modern distributed storage systems, node failures are common and single-node failures
occur most frequently among all failures \cite{2013XORing, 2013A}. It is important to
repair the failed node with the {\em repair bandwidth} defined as the total amount of
symbols downloaded from other surviving nodes as small as possible \cite{dimakis2010}.
Recently many constructions of $(n,k,m)$ MDS array codes to achieve the minimum repair bandwidth
$\frac{dm}{d-k+1}$ from $d$ surviving nodes have been proposed in \cite{rashmi2011,hou2016,li2018,2020Toward,2020Binary,ye2017}, all the high-code-rate
(i.e.,$k/n > 0.5$) MDS array codes with the minimum repair bandwidth need an
exponential sub-packetization level in parameters $n$ and $k$ \cite{2018A}.
It is practical important to design high-code-rate MDS array codes with repair bandwidth
as small as possible, for a given small sub-packetization level. HashTag Erasure Codes (HTEC)
\cite{2018HashTag} is a high-code-rate that have efficient repair method for data nodes with
sub-packetization level $2 \leq m \leq r^{\lceil\frac{k}{r}\rceil}$, however no efficient
repair method for parity nodes and the required field size
should be large enough to keep the MDS property.

Piggybacking codes which were first proposed by Rashmi \emph{et al.} in \cite{2017A}
are an important class of MDS array codes that have both low repair bandwidth for
single-node failure and low sub-packetizaiton level. The central idea of piggybacking
codes is creating $m$ instances of RS codes as base codes and designing ingenious
{\em piggyback function} (i.e., a linear combination of some selected symbols in some instances)
which will be added to other instances. Many follow-up piggybacking codes
\cite{2018Repair,2019AnEfficient,article,2021piggyback,2021piggybacking,2022arXiv220509659W} have been
proposed to reduce the repair bandwidth.

In this paper, we present two constructions of piggybacking codes that have lower
repair bandwidth than the existing piggybacking codes for the same parameters. We summarize
the main contributions as follows.
\begin{enumerate}
\item First, we propose first piggybacking codes for $r\geq 4$ and $m\leq r$. We show that our
first piggybacking codes have lower repair bandwidth for any single-node failure than all the
existing piggybacking codes for the evaluated parameters $r = 8,9$, $m = 6$ and $k = 30,31,\ldots,100$.
\item Second, we propose second piggybacking codes by jointly designing piggyback function
for data nodes repair and {\em transformation function} for parity nodes repair, where the sub-packetization
level is a multiple of the number of parity nodes. The proposed second piggybacking codes
have the lowest repair bandwidth for any single-node failure among all the existing
piggybacking codes for the evaluated parameters $\frac{k}{n} = 0.75, 0.8, 0.9$ and $r \geq 4$.
\end{enumerate}

Note that a parallel work \cite{2022arXiv220509659W} also designs piggybacking codes
to obtain low repair bandwidth for $m\leq r$. The differences of codes \cite{2022arXiv220509659W}
and our first piggybacking codes are as follows. We design the piggyback function
by considering the repair bandwidth reduction for both data nodes and parity nodes.
While in \cite{2022arXiv220509659W}, the piggyback functions for data nodes and
parity nodes are respectively designed. Because of the above difference, our codes
have a slightly lower repair bandwidth than that of codes in \cite{2022arXiv220509659W}.
Please refer to Section \ref{sec:com-1} for the comparison.

The main differences between our first piggybacking codes and codes in our conference version
\cite{Shi2206:Piggybacking} are of two-folds.
First, our first piggybacking codes can support flexible sub-packetization $m$,
i.e., $2\leq m\leq r$, while codes in \cite{Shi2206:Piggybacking} only suitable for
$m=r$. Second, the piggyback structure of our first piggybacking codes can be jointly
designed with the proposed transformation function, while not for codes
in \cite{Shi2206:Piggybacking}.


Our second piggybacking codes is partially inspired by the generic transformation in
\cite{li2018}. The difference is that new MDS array codes with exponential sub-packetization
level and optimal repair for any single-node failure can be obtained in \cite{li2018}
by recursively applying the transformation for MDS codes, while we use the transformation
idea to design the transformation function for parity nodes in order to reduce the repair bandwidth
in the meanwhile keeping the low repair bandwidth for data nodes.
Note that it is not natural to obtain repair bandwidth reduction when we design piggyback functions
for data nodes and transformation functions for parity nodes, since both piggyback function
and transformation function are added in the same parity symbol. We need to carefully design
the two functions to achieve lower repair bandwidth. Moreover, we can't design transformation
functions for parity nodes of piggybacking codes with invertible transformation
functions \cite{2017A,2019AnEfficient,Shi2206:Piggybacking}, because the invertible
transformation structure will be destroyed
if the transformation idea is employed for parity nodes.

The rest of the paper is organized as follows.
Section~\ref{sec:framework1} gives the construction for the first piggybacking codes.
Section~\ref{sec:repair1} presents the repair method for the first piggybacking codes.
Section~\ref{sec:framework2} gives the construction for the second piggybacking codes.
Section~\ref{sec:repair2} presents the repair method for the second piggybacking codes.
Section~\ref{sec:comparison} evaluates the repair bandwidth for the proposed
piggybacking codes and the existing related piggybacking codes.
Section~\ref{sec:con} concludes the paper.

\section{Construction of the First Piggybacking Codes}\label{sec:framework1}

Our first piggybacking codes can be represented by an $n\times m$ array,
where the $m$ symbols in each row are stored in a node and $m\leq r=n-k$.
We label the index of the $n$ rows in the array from $1$ to $n$ and the index of
the $m$ columns from 1 to $m$. Let $\{\boldsymbol{a}_i = (a_{i,1},a_{i,2},\ldots,a_{i,k})^T\}^{m}_{i = 1}$ be $m$ columns of the $k \times m$ data symbols and $(a_{i,1},a_{i,2},\ldots,a_{i,k},f_1(\boldsymbol{a}_i),\ldots,f_r(\boldsymbol{a}_i))^T$ be
codeword $i$ of the $\left(n,k\right)$ MDS codes over $\mathbb{F}_q$, where
$f_j(\boldsymbol{a}_i)$ is the parity symbol $j$ in codeword $i$,
$i = 1,2,\ldots,m$ and $j=1,2,\ldots,r$.

We divide $n$ nodes into $L$ disjoint subsets $\Phi_1,\Phi_2, \ldots, \Phi_L$, where
$1\leq L <m$. Each of the first $n-\lfloor\frac{n}{L}\rfloor L$ subsets has size
$\lceil\frac{n}{L}\rceil$ and each of the last $(\lfloor\frac{n}{L}\rfloor+1)L-n$
subsets has size $\lfloor\frac{n}{L}\rfloor$, i.e., $|\Phi_i|=\lceil\frac{n}{L}\rceil$ for $i=1,2,\ldots,n-\lfloor\frac{n}{L}\rfloor L$ and $|\Phi_i|=\lfloor\frac{n}{L}\rfloor$ for $i=n-\lfloor\frac{n}{L}\rfloor L+1,\ldots,L$.
We can check that
$$\lceil\frac{n}{L}\rceil(n-\lfloor\frac{n}{L}\rfloor L)+
\lfloor\frac{n}{L}\rfloor((\lfloor\frac{n}{L}\rfloor+1)L-n)=n,$$
i.e., the $n$ nodes $\{1,2,\ldots,n\}$ are partitioned by the $L$ disjoint subsets.
In this paper, we consider high-code-rate and suppose that the $r$ parity nodes are
in the subset $\Phi_L$, i.e.,
$$\lfloor\frac{n}{L}\rfloor\geq r.$$
For example, when $k=6$,
$r = 5$  and $L=2$, the $n = 11$ nodes are divided into $L=2$ subsets $\Phi_1=\{1,2,3,4,5,6\}$
and $\Phi_2=\{7,8,9,10,11\}$.

Recall that $|\Phi_i|\geq \lfloor\frac{n}{L}\rfloor\geq r\geq m$ and $\Phi_i$ contains $|\Phi_i|$ nodes, where $i = 1,2,\ldots,L$.
For $i = 1,2,\ldots,L$, define the first $m-i$ symbols of each node in $\Phi_i$ as
Protect Symbols (PS) which will be added to some parity symbols as piggyback function.
In the following, we present a method of designing the $(r-1)L$ piggyback functions that
are added to the $(r-1)L$ parity symbols such that the number of PS used in computing
each piggyback function as average as possible.

The total number of PS in $\Phi_i$ is $p_i = |\Phi_i|(m - i)$.
For $i = 1,2,\ldots,L$, denote the PS in column $j$ with $j=1,2,\ldots,m-i$ in row
$(\sum_{\alpha=1}^{i-1}|\Phi_{\alpha}|)+\ell$ with $\ell=1,2,\ldots,|\Phi_i|$ as $t_{i,(\ell - 1)(m-i)+j}$.
For example, when $i = 1$, we have
\begin{align*}
&(t_{1,1},t_{1,2},\ldots,t_{1,m - 1})=(a_{1,1},a_{2,1},\ldots,a_{m-1,1}),\\
&(t_{1,m},t_{1,m+1},\ldots,t_{1,p_1})=(a_{1,2},a_{2,2},\ldots,a_{m-1,|\Phi_1|}).
\end{align*}

For $1 \leq \alpha \leq r-1$ and $1 \leq \beta < L $, let
$s_{\alpha,\beta}=1+\lfloor \frac{p_{\beta}-\alpha}{r-1}\rfloor$ and
define the piggyback function
$g(\alpha, \beta)$ as,
\begin{equation}\label{eq1}
	g(\alpha, \beta) = \sum_{\ell = 1}^{s_{\alpha,\beta}} t_{\beta, (\ell - 1)(r-1)+\alpha}.
\end{equation}
For $1 \leq \alpha \leq r-1$ and $\beta = L$, we define the piggyback function $g(\alpha, \beta)$ as,
\begin{eqnarray}\label{eq2}
	g(\alpha,\beta)
	&=& \sum_{\ell = 1}^{s_{\alpha,\beta}} t_{\beta, (\ell-1)(r-1)+\alpha-((m-L)r\bmod (r-1))} \\ \nonumber
	&+& \sum_{x=1}^{r}\sum_{y=1}^{m-L}f_x(\boldsymbol{a}_y)\big(\mathbf{I}[x+y=\alpha + 1]
	+\mathbf{I}[x+y-(r-1) = \alpha + 1]\big), \nonumber
\end{eqnarray}
where $\mathbf{I}[\cdot]$ is characteristic function (i.e., $\mathbf{I}[A] = 1$ if $A$ is true,
otherwise $\mathbf{I}[A] = 0$) and
$$s_{\alpha,L} = 1 + \lfloor\frac{(|\Phi_L| - r)(m-L) - \alpha +(m-L)r \bmod (r-1)}{r-1}\rfloor.$$
Notice that we let $t_{L,\ell} = 0$ if $\ell \leq 0$ in Eq. \eqref{eq2}.
We add the piggyback function $g(\alpha, \beta)$ to the symbol in row $\alpha+k+1$ and column $m+1-\beta$.

We denote the above designed piggybacking codes as $\mathcal{C}_1(n, k, m,L)$.
Fig. \ref{fig:1} shows the construction structure of $\mathcal{C}_1(n, k, m,L)$.
When $m=r$, the piggyback structure of our $\mathcal{C}_1(n, k, m=r,L)$ is quite similar
the conference version \cite{Shi2206:Piggybacking}. One difference is that
invertible transformation is used in codes \cite{Shi2206:Piggybacking}, while not
in our $\mathcal{C}_1(n, k, m=r,L)$. This is why we can jointly design piggyback
function and transformation function for repair bandwidth reduction (please refer to
Section \ref{sec:framework2} for details), however the jointly design is not suitable for codes \cite{Shi2206:Piggybacking}.

\begin{figure*}[htbp]
	\centering
	\includegraphics[width=1\textwidth]{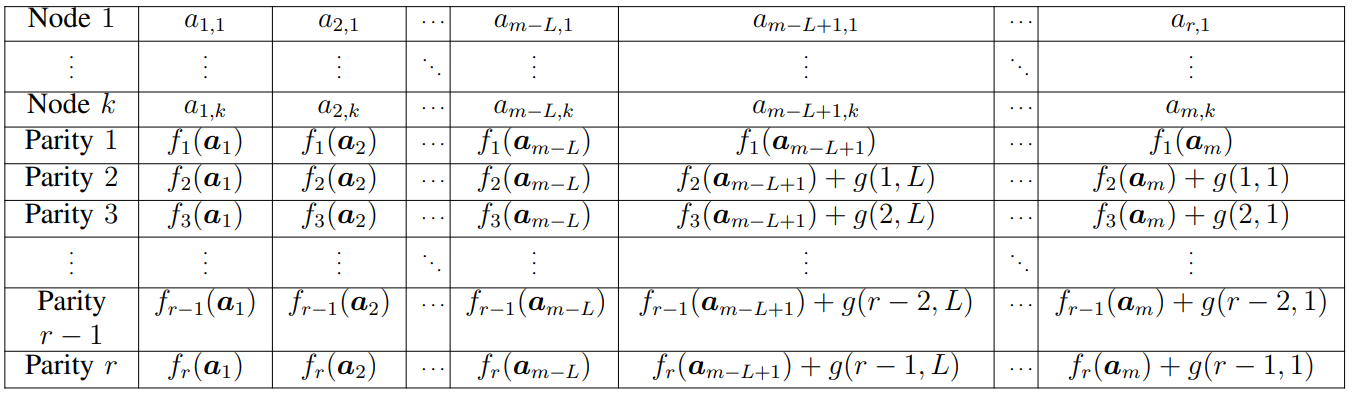}
	\caption{The construction structure of $\mathcal{C}_1(n,k,m,L)$.}
	\label{fig:1}
\end{figure*}

\begin{example}\label{ex2}
Consider the example of $\mathcal{C}_1(11, 6, 4,2)$, which is shown in Fig. \ref{fig:2}.
We divide the $n=11$ nodes into $L = 2$ disjoint subsets, $\Phi_1 = \{1,2,3,4,5,6\}$ and
$\Phi_2 = \{7,8,9,10,11\}$. By Eq. \eqref{eq1} and Eq. \eqref{eq2}, we have
	\begin{eqnarray}\nonumber
		g(1,1) &=& a_{1,1} + a_{2,2} + a_{3,3} + a_{1,5} + a_{2,6},\\ \nonumber
		g(2,1) &=& a_{2,1} + a_{3,2} + a_{1,4} + a_{2,5} + a_{3,6},\\ \nonumber
		g(3,1) &=& a_{3,1} + a_{1,3} + a_{2,4} + a_{3,5},\\ \nonumber
		g(4,1) &=& a_{1,2} + a_{2,3} + a_{3,4} + a_{1,6},\\ \nonumber
		g(1,2) &=& f_1(\boldsymbol{a}_1) + f_4(\boldsymbol{a}_2) + f_5(\boldsymbol{a}_1),\\ \nonumber
		g(2,2) &=& f_1(\boldsymbol{a}_2) + f_2(\boldsymbol{a}_1) + f_5(\boldsymbol{a}_2),\\ \nonumber
		g(3,2) &=& f_2(\boldsymbol{a}_2) + f_3(\boldsymbol{a}_1),\\ \nonumber
		g(4,2) &=& f_3(\boldsymbol{a}_2) + f_4(\boldsymbol{a}_1). \nonumber
	\end{eqnarray}	
	
	\begin{figure}[htbp]
		\centering
		\includegraphics[width=0.8\textwidth]{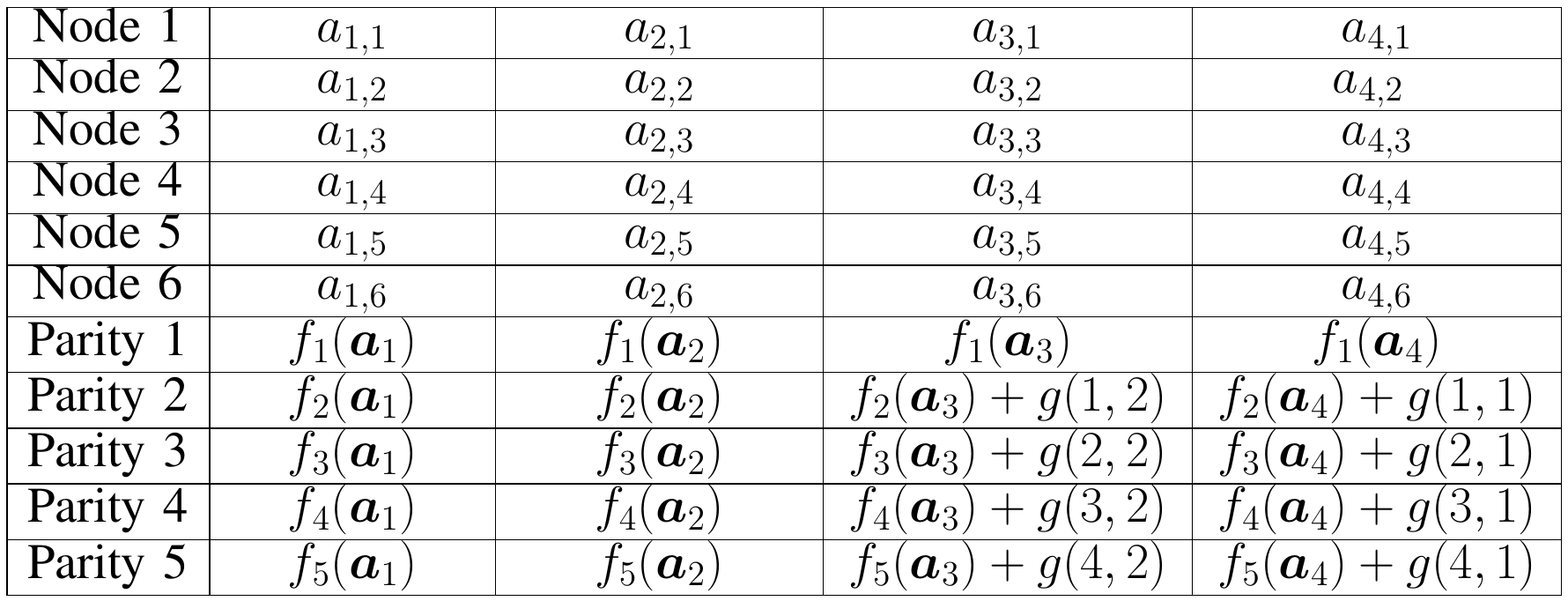}
		\caption{The example of $\mathcal{C}_1(n = 11, k = 6, m = 4,L=2)$.}
		\label{fig:2}
	\end{figure}
\end{example}

According to the above definition of the piggyback function, we can easily know that
any two data symbols in the same row are not used to
compute one piggyback function because of $m \leq r$.
In the next lemma, we show that this is also true for parity nodes.

\begin{lemma}
When $r \geq 4$, we have,
\begin{description}

\item[(i)]
Any two parity symbols in the same row are not used to compute the same
piggyback function.
\item[(ii)] Any parity symbol in a row used in computing a piggyback function is
not in the same row of the piggyback function.
\end{description}
\label{lm:piggy}
\end{lemma}

\begin{proof}
Since $|\Phi_i|\geq r\geq m$, the piggyback functions which are computed from parity symbols are
$g(\alpha,L)$ with $1 \leq \alpha \leq r-1$.

Consider the first claim. Suppose that two parity symbols
$f_{x}(\boldsymbol{a}_{y_1})$ and $f_{x}(\boldsymbol{a}_{y_2})$ that are in the same row
which are used to compute
the piggyback function $g(\alpha,L)$, where $x\in\{1,2,\ldots,r\}$ and
$1\leq y_1 < y_2\leq m-L $, then
	\begin{eqnarray}\nonumber
		& x +y_1 = \alpha+1, \\ \nonumber
		& x +y_2 -r+1  = \alpha+1.  \nonumber
	\end{eqnarray}
We have that $y_2 - y_1 = r-1$, which contradicts with $r\geq m$ and $L\geq 1$.

Consider the second claim.
Suppose that the parity symbol $f_x(\boldsymbol{a}_y)$ is used to compute the piggyback
function $g(\alpha, \beta)$ in the same row, we have $x = \alpha + 1$, where $x\in\{1,2,\ldots,r\}$,
$y\in\{1,2,\ldots,m-L\}$, $x+y=\alpha+1$ or $x+y-(r-1) = \alpha+1$.
We can obtain that $y = 0$ or $y = r - 1$, which contradicts with $1 \leq y \leq m-L$,
$r\geq m$ and $L\geq 1$.
\end{proof}

By Lemma \ref{lm:piggy}, we can repair any single-node failure by employing the
piggyback functions to reduce repair bandwidth and we present the repair method
in the next section.

\section{Repair Method of Codes $\mathcal{C}_1(n, k, m,L)$}\label{sec:repair1}

In this section, we present the repair method for any single-node failure of the proposed code $\mathcal{C}_1(n, k, m,L)$ and show the repair bandwidth.

\subsection{Repair Method for Data Nodes} \label{py1dt}

Suppose that node $t \in \Phi_i$ fails, where $t \in \{1, 2,\ldots, k\}$ and $i \in \{1,2,\ldots,L\}$, the repair method is as follows.

\begin{enumerate}
	\item We download $ki$ symbols in the last $i$ columns of the first $k + 1$ rows except row $t$ to recover the symbols  $a_{m-i+1,t},\ldots,a_{m,t}$ and $f_x(\boldsymbol{a}_i)$ for $ x = 2,3,\ldots,r$.
	
	\item We download the parity symbols of which the corresponding piggyback functions
containing symbols of node $t$ and symbols contained by these piggyback functions except
symbols of node $t$, together with $f_x(\boldsymbol{a}_i)$
for $x = 2, 3, \ldots, r$, we can recover $a_{1,t}, a_{2,t},\ldots, a_{m-i,t}.$
	
\end{enumerate}

Consider the code $\mathcal{C}_1(11, 6, 4,2)$ in Example \ref{ex2}, we have $L = 2$
disjoint subsets $\Phi_1 = \{1,2,3,4,5,6\}$ and $\Phi_2 = \{7,8,9,10,11\}$.

Suppose that node $1$ fails, we can recover the five symbols
\[
a_{4,1},f_2(\boldsymbol{a}_4),f_3(\boldsymbol{a}_4),f_4(\boldsymbol{a}_4),f_5(\boldsymbol{a}_4)
\]
by downloading the following six symbols
\[
a_{4,2}, a_{4,3}, a_{4,4}, a_{4,5}, a_{4,6},f_1(\boldsymbol{a}_4).
\]
Note that the erased three symbols $a_{1,1}$, $a_{2,1}$ and $a_{3,1}$ are used in computing the three
piggyback functions $g(1,1)$, $g(2,1)$ and $g(3,1)$, respectively. We download the three parity
symbols
$$f_2(\boldsymbol{a}_4)+g(1,1), f_3(\boldsymbol{a}_4) + g(2,1), f_4(\boldsymbol{a}_4)+ g(3,1)$$
and the following symbols
$$a_{2,2},a_{3,3},a_{1,5},a_{2,6},a_{3,2},a_{1,2},a_{2,5},a_{3,6},a_{1,3},a_{2,4},a_{3,5}$$
that are used in computing three piggyback functions $g(1,1)$, $g(2,1)$ and $g(3,1)$,
together with $f_2(\boldsymbol{a}_4),f_3(\boldsymbol{a}_4),f_4(\boldsymbol{a}_4)$, to recover
the three symbols $a_{1,1}$, $a_{2,1}$ and $a_{3,1}$. The repair bandwidth of node $1$ is $20$ symbols.

We can repair each of the other data nodes similarly and we can calculate that the repair
bandwidth of each node in $\{2,5,6\}$ is 20 symbols, the repair bandwidth of each node in
$\{3,4\}$ is 19 symbols.

\subsection{Repair Method for Parity Nodes} \label{py1pt}
Suppose that node $t$ fails, where $t \in \{k+1, k+2,\ldots, k+r\}$ and $t\in \Phi_L$,
the repair method is given as follows.

\begin{enumerate}
	\item We download $kL$ symbols in the last $L$ columns of the first $k$ rows to recover $f_{t-k}(\boldsymbol{a}_{j})$ with $j = m-L+1, m-L+2,\ldots,m$ and $f_{x}(\boldsymbol{a}_{m-L+1})$
with $x = 1,2,\ldots,t-k-1,t-k+1,\ldots,r$.
	\item Note that the $m-L$ erased symbols
$f_{t-k}(\boldsymbol{a}_1),f_{t-k}(\boldsymbol{a}_2),\ldots,f_{t-k}(\boldsymbol{a}_{m-L})$
are used in computing the $m-L$ piggyback functions $g((t-k-1+s) \bmod (r-1), L)$ with
$s=1,2,\ldots, m-L$.
We can recover $m-L$ symbols
$f_{t-k}(\boldsymbol{a}_1),f_{t-k}(\boldsymbol{a}_2),\ldots,f_{t-k}(\boldsymbol{a}_{m-L})$
by downloading the $m-L$ symbols
$f_{(t-k+s)\bmod (r-1)}(\boldsymbol{a}_{m-L+1}) + g((t-k-1+s)\bmod (r-1), L)$ with
$s=1,2,\ldots, m-L$ and the symbols used in computing the piggyback functions
$g((t-k-1+s) \bmod (r-1), L)$ with $s=1,2,\ldots, m-L$.	
	\item We can recover the last $L$ symbols in node $t$ by downloading the symbols used in
computing the $L$ piggyback functions added in node $t$, together with the computed symbols
in the first step.
\end{enumerate}

Continue the code in Example \ref{ex2}. Suppose that parity $1$ fails.
According to the above repair method, in the first step, we can compute the four symbols
$f_1(\boldsymbol{a}_3), f_1(\boldsymbol{a}_4),f_2(\boldsymbol{a}_3), f_3(\boldsymbol{a}_3)$ by
downloading the following 12 symbols
	\begin{eqnarray}\nonumber
		&a_{3,1},a_{3,2},a_{3,3},a_{3,4},a_{3,5},a_{3,6}, \\ \nonumber
		&a_{4,1},a_{4,2},a_{4,3},a_{4,4},a_{4,5},a_{4,6}.  \nonumber
	\end{eqnarray}
By the second step, we can recover the erased two symbols $f_1(\boldsymbol{a}_1),f_1(\boldsymbol{a}_2)$
by downloading $$f_4(\boldsymbol{a}_2), f_5(\boldsymbol{a}_1), f_2(\boldsymbol{a}_1),f_5(\boldsymbol{a}_1),f_2(\boldsymbol{a}_3)+g(1,2), f_3(\boldsymbol{a}_3) + g(2,2),$$ since
\begin{align*}
g(1,2) =& f_1(\boldsymbol{a}_1) + f_4(\boldsymbol{a}_2) + f_5(\boldsymbol{a}_1),\\
g(2,2) =& f_1(\boldsymbol{a}_2) + f_2(\boldsymbol{a}_1) + f_5(\boldsymbol{a}_2).
\end{align*}
There is no piggyback function added in the parity 1 and we have recovered the last
two symbols in parity 1, the third step is not necessary in repairing parity 1.

Suppose that parity $2$ fails.
By the first step, we can compute the four symbols $f_2(\boldsymbol{a}_3), f_2(\boldsymbol{a}_4)$, $f_3(\boldsymbol{a}_3), f_4(\boldsymbol{a}_3)$ by downloading the following 12 symbols
	\begin{eqnarray}\nonumber
		&a_{3,1},a_{3,2},a_{3,3},a_{3,4},a_{3,5},a_{3,6}, \\ \nonumber
		&a_{4,1},a_{4,2},a_{4,3},a_{4,4},a_{4,5},a_{4,6}.  \nonumber
	\end{eqnarray}	
By the second step, we recover two symbols $f_2(\boldsymbol{a}_1),f_2(\boldsymbol{a}_2)$ by
downloading the following symbols
$$f_1(\boldsymbol{a}_2), f_5(\boldsymbol{a}_2), f_3(\boldsymbol{a}_1),f_3(\boldsymbol{a}_3)+g(2,2), f_4(\boldsymbol{a}_3) + g(3,2).$$
By the third step, we can recover the two symbols
$f_2(\boldsymbol{a}_3) +g(1,2),f_2(\boldsymbol{a}_4)+g(1,1)$
by downloading the symbols $$a_{1,1},a_{2,2},a_{1,5},a_{2,6},f_2(\boldsymbol{a}_2),f_3(\boldsymbol{a}_1),$$ together with $a_{3,3}$ and $f_2(\boldsymbol{a}_3),f_2(\boldsymbol{a}_4)$.

We can calculate that the repair bandwidth of parity 2 is $23$ symbols. Similarly, we can
show that the repair bandwidth of each node in parity $\{3,5\}$ is $24$ symbols and the repair bandwidth of
parity $4$ is $23$ symbols.

\subsection{Average Repair Bandwidth Ratio of $\mathcal{C}_1(n,k,m,L)$}
In the following, we analyse the repair bandwidth of our codes $\mathcal{C}_1(n,k,m,L)$.
We define the {\em average repair bandwidth ratio} of all nodes  as the ratio of the average repair bandwidth of $n$ nodes to the number of data symbols $kr$.

\begin{lemma}\label{lemma3}
	If $L$ is a factor of $n$, the lower bound $\gamma_{min}^{all}$ and the upper bound
$\gamma_{max}^{all}$ of the average repair bandwidth ratio of all nodes of $\mathcal{C}_1(n, k, m, L)$ is
	\begin{eqnarray}\nonumber
		&\gamma_{min}^{all} = \frac{L+1}{2m} + \frac{(k+r)(m^2-m(L+1)+\frac{(L+1)(2L+1)}{6})}{Lmk(r-1)} + \frac{m-L}{km}, \\ \nonumber
		&\gamma_{max}^{all} = \gamma_{min}^{all} + \frac{L(r-1)^2}{4(k+r)mk}. \nonumber
	\end{eqnarray}
\end{lemma}

\begin{proof}
Denote the number of symbols used in computing the piggyback function $g(\alpha,r-m+\beta)$
as $n_{\alpha,\beta}$, where $\alpha = 1,2,\ldots,r-1$ and $\beta = 1,2,\ldots,L$.
According to the definition of the piggyback function in Eq. \eqref{eq1} and Eq. \eqref{eq2},
we can know that $(n_{\alpha_{1},\beta}-n_{\alpha_{2},\beta})^2 \in \{0,1\}$ for
$\alpha_1 \neq \alpha_2\in\{1,2,\ldots,r-1\}$ and $\beta = 1,2,\ldots,L$.

According to the repair methods in Section \ref{py1dt} and \ref{py1pt}, we need to
download $\sum_{\beta =1}^{L}|\Phi_\beta|k\beta$ symbols in repairing each of the
$n$ nodes in first step, download $\sum_{\beta=1}^{L}\sum_{\alpha=1}^{r-1}n_{\alpha,\beta}^2$
symbols in repairing each of the $n$ nodes in the second step and download
$(m-L)(k+r)$ symbols in repairing each of the $r$ parity nodes in the third step in total.
Therefore, the average repair bandwidth ratio for all nodes $\gamma^{all}$ is
$$\gamma^{all} = \frac{\sum_{\beta=1}^{L}(|\Phi_\alpha|k\beta +
\sum_{\alpha=1}^{r-1}n_{\alpha,\beta}^2) + (m-L)(k+r)}{(k+r)mk}.$$
Since the equation
$$(\sum_{\alpha=1}^{r-1}n_{\alpha,\beta})^2 +
\sum_{\alpha_1\neq \alpha_2\in\{1,2,\ldots,r-1\}}(n_{\alpha_1,\beta}-n_{\alpha_2,\beta})^2
= (r-1)\sum_{\alpha=1}^{r-1}n^2_{\alpha,\beta}$$
holds, we can obtain that
$$\sum_{\alpha=1}^{r-1}n_{\alpha,\beta}^2 =
\frac{(\sum_{\alpha=1}^{r-1}n_{\alpha,\beta})^2 +
\sum_{\alpha_1\neq \alpha_2\in\{1,2,\ldots,r-1\}}(n_{\alpha_1,\beta}-n_{\alpha_2,\beta})^2}{r-1}$$
for $\beta = 1,2,\ldots,L$.

Let $v_\beta = |\Phi_\beta|(m-\beta) - \lfloor\frac{|\Phi_\beta|(m-\beta)}{r-1}\rfloor (r-1)$,
where $\beta = 1,2,\ldots,L$. Because of $\sum_{\alpha=1}^{r-1}n_{\alpha,\beta} =
|\Phi_\beta|(m-\beta)$ and $\sum_{\alpha_1,\alpha_2=1,\alpha_1\neq \alpha_2}^{r-1}
(n_{\alpha_1,\beta}-n_{\alpha_2,\beta})^2 = v_\beta(r-1-v_\beta)$, we can
calculate that
$$\sum_{\alpha=1}^{r-1}n_{\alpha,\beta}^2 =
\frac{(|\Phi_\beta|(m-\beta))^2 + v_\beta(r-1-v_\beta)}{r-1}.$$
Therefore, we can get  $0 \leq v_\beta(r-1-v_\beta)\leq (\frac{r-1}{2})^2$ and
further obtain the lower bound and upper bound in the lemma.
\end{proof}

By Lemma \ref{lemma3}, we can know $|\gamma_{min}^{all} -
\gamma^{all}| \leq |\gamma_{max}^{all} - \gamma_{min}^{all}| = \frac{L(r-1)^2}{4(k+r)mk}$.
When $k >> r$, we have $\frac{L(r-1)^2}{4(k+r)mk} \rightarrow 0$. Therefore, we have
$\gamma^{all}= \gamma_{min}^{all}$ when $k >> r$.

\begin{lemma}\label{lemma4}
When $L$ is a factor of $n$ and $k >> r$, the minimum value of the average repair
bandwidth ratio $\gamma^{all}$ of  $\mathcal{C}_1(n, k, m, L)$ is achieved when $L = \sqrt{\frac{6m^2-6m+1}{3r-1}}$.
\end{lemma}
\begin{proof}
When $k >> r$, we have $\gamma^{all} =\gamma^{all}_{min} $ . Then we can get
	$$\gamma^{all} = \frac{L+1}{2m} + \frac{m^2L-mL^2-mL+\frac{1}{3}L^3 +
\frac{1}{2}L^2+\frac{L}{6}}{L^2m(r-1)}.$$
We can calculate that
$$\frac{\partial \gamma^{all}}{\partial L} = \frac{1}{2m} +
\frac{\frac{1}{3} - \frac{m^2-m+\frac{1}{6}}{L^2} }{m(r-1)} = 0,$$
and further obtain $$L = \sqrt{\frac{6m^2-6m+1}{3r-1}}.$$
If $L > \sqrt{\frac{6m^2-6m+1}{3r-1}}$, then $\frac{\partial \gamma^{all}}{\partial L} > 0$;
if $L < \sqrt{\frac{6m^2-6m+1}{3r-1}}$, then $\frac{\partial \gamma^{all}}{\partial L} < 0$;
if $L = \sqrt{\frac{6m^2-6m+1}{3r-1}}$, then  $\frac{\partial \gamma^{all}}{\partial L} = 0$.
Therefore, when $L = \sqrt{\frac{6m^2-6m+1}{3r-1}} $, $\gamma^{all}$ achieves the minimum value.
\end{proof}

Since $L$ is a positive integer, we take $L = \lceil\sqrt{\frac{6m^2-6m+1}{3r-1}}\rceil$ or
$L = \lfloor\sqrt{\frac{6m^2-6m+1}{3r-1}}\rfloor $ to achieve the minimum repair bandwidth.

\section{The Second Piggybacking Codes}\label{sec:framework2}

In this section, we present construction of the second piggybacking codes that have
lower repair bandwidth than all the existing piggybacking codes.

The second piggybacking codes is an $n\times m$ array, where $m=sr=s(n-k)$ and $2\leq s\leq r$.
We label the index of the $n$ rows from $1$ to $n$ and the index of
the $m$ columns from 1 to $m$. The first $k$ nodes are data nodes that store data symbols
and the last $r$ nodes are parity nodes that store parity symbols. We divide the $k$ data
nodes into $L$ disjoint subsets $\Phi_1,\Phi_2, \ldots, \Phi_L$, where
$1\leq L< s$. Each of the first $k-\lfloor\frac{k}{L}\rfloor L$ subsets has size
$\lceil\frac{k}{L}\rceil$ and each of the last $(\lfloor\frac{k}{L}\rfloor+1)L-k$
subsets has size $\lfloor\frac{k}{L}\rfloor$, i.e., $|\Phi_i|=\lceil\frac{k}{L}\rceil$ for $i=1,2,\ldots,k-\lfloor\frac{k}{L}\rfloor L$ and $|\Phi_i|=\lfloor\frac{k}{L}\rfloor$ for $i=k-\lfloor\frac{k}{L}\rfloor L+1,\ldots,L$.
We divide the $sr$ columns into $s$ disjoint subsets $W_{\ell}=\{(\ell-1)r+1,(\ell-1)r+2,\ldots,\ell r\}$,
where $\ell=1,2,\ldots,s$.
The construction of the second piggybacking codes is as follows.

\begin{enumerate}
\item We create $m=sr$ instances of $\left(n,k\right)$ MDS code. Let
$(a^{(\ell)}_{i,1},a^{(\ell)}_{i,2},\ldots,a^{(\ell)}_{i,k},
f_1(\boldsymbol{a}^{(\ell)}_i),\ldots,f_r(\boldsymbol{a}^{(\ell)}_i))^T$ be an instance
of $\left(n,k\right)$ MDS code over $\mathbb{F}_{q}$ which is in column $(\ell-1)r +i$ of the $n\times sr$ array,
where $\boldsymbol{a}^{(\ell)}_{i} = (a^{(\ell)}_{i,1},a^{(\ell)}_{i,2},\ldots,a^{(\ell)}_{i,k})^T$ are data symbols,
$\ell = 1,2,\ldots,s$ and $i = 1,2,\ldots,r$.
\item We take $i-1$ cyclic-shift for each column of the $r$ parity symbols
$(f_1(\boldsymbol{a}^{(\ell)}_i),f_2(\boldsymbol{a}^{(\ell)}_i),\ldots,
f_{r}(\boldsymbol{a}^{(\ell)}_i))^T$ to obtain
$$(f_{((1-i)\bmod r)+1}(\boldsymbol{a}^{(\ell)}_i),f_{((2-i)\bmod r)+1}(\boldsymbol{a}^{(\ell)}_i),\ldots,f_r(\boldsymbol{a}^{(\ell)}_i),f_1(\boldsymbol{a}^{(\ell)}_i),
f_2(\boldsymbol{a}^{(\ell)}_i),\ldots,f_{r+1-i}(\boldsymbol{a}^{(\ell)}_i))^T,$$
where $\ell = 1,2,\ldots,s$ and
$i = 1,2,\ldots,r$.
\item For $i= 1,2,\ldots,L$, we define the first $(s-i)r$ symbols of the $|\Phi_i|$
nodes in $\Phi_i$ as Protect Symbols (PS). We can calculate that the total number of
PS in $\Phi_i$ is $p_i = |\Phi_i|(s - i)r$ and the total number of PS in all $k$ data nodes are
$\sum_{i=1}^{L}p_i =\sum_{i=1}^{L} |\Phi_i|(s - i)r$. For $i = 1,2,\ldots,L$, denote the PS in
column $j$ with $j=1,2,\ldots,(s-i)r$ in row $(\sum_{\alpha=1}^{i-1}|\Phi_{\alpha}|)+\ell$
with $\ell=1,2,\ldots,|\Phi_i|$ as $t_{i,(\ell - 1)(s-i)r+j}$. For example, when $i = 1$,
we have
\begin{align*}
&(t_{1,1},t_{1,2},\ldots,t_{1,(s-1)r})=(a^{(1)}_{1,1}, a^{(1)}_{2,1},\ldots,a^{(s-1)}_{r,1}),\\
&(t_{1,(s-1)r+1},\ldots,t_{1,p_1})=(a^{(1)}_{1,2},\ldots,a^{(s-1)}_{r,|\Phi_1|}).
\end{align*}
We define $(r-1)rL$ piggyback functions which will be added to the $(r-1)rL$ parity symbols in the
last $L$ subsets $W_i$ with $i=s-L+1,s-L+2,\ldots,s$ such that the number of PS used in
computing each piggyback function as average as possible. For $i=1,2,\ldots,L$ and
$j=1,2,\ldots,r(r-1)$, we define the piggyback function $g(j,s-i + 1)$ as
\begin{equation}\label{gx}
g(j,s-i + 1) = \sum_{\ell = 1}^{s_{i,j}} t_{i,(\ell-1)(r-1)r + j},
\end{equation}
where $s_{i,j} = 1 +\lfloor\frac{p_i-j}{r(r-1)}\rfloor$ and the $r(r-1)L$ piggyback functions
are added to the $r(r-1)L$ parity symbols in the last $L$ subsets $W_{\ell}$ with $\ell = s-L+1,s-L+2,\ldots,s$,
which is shown in Fig. \ref{pr}.
For notational convenience, denote the symbol in column $(\ell-1)r+i$ and row $j$ of the obtained
$n\times sr$ array as $f'_{x-k}(\boldsymbol{a}^{(\ell)}_i)$, where $j=k+1,k+2,\ldots,k+r$,
$\ell = 1,2,\ldots,s$ and $i = 1,2,\ldots,r$.
	\begin{figure*}[htbp]
		\centering
		\includegraphics[width=1\textwidth]{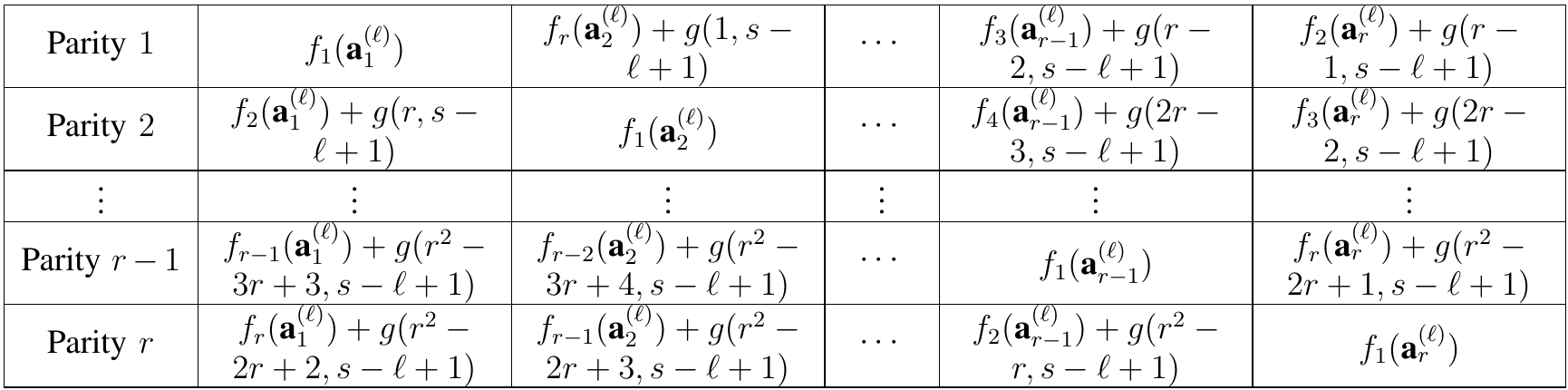}
		\caption{Piggyback functions added to the symbols in $W_\ell$, $\ell = s-L+1,s-L+2,\ldots,s$.}
		\label{pr}
	\end{figure*}
\item We replace the symbol $f'_{x}(\boldsymbol{a}^{(\ell)}_i)$ in column $(\ell-1)r+i$ and row $x+k$
by the {\em transformation function} $f''_x(\boldsymbol{a}^{(\ell)}_i)$ as follows,
	\begin{equation}
		f''_x(\boldsymbol{a}^{(\ell)}_i) =
		\begin{cases}
			f'_{x}(\boldsymbol{a}^{(\ell)}_i) + f'_i(\boldsymbol{a}^{(\ell)}_{x})  & \text{if} \quad x < i, \\
			\theta_{x,i}\cdot f'_{i}(\boldsymbol{a}^{(\ell)}_x) + f'_x(\boldsymbol{a}^{(\ell)}_{i})  & \text{if} \quad x > i, \\
			f'_{x}(\boldsymbol{a}^{(\ell)}_i) & \text{if} \quad x = i,
		\end{cases}
\label{eq:trans}	
\end{equation}
where $\theta_{x,i}\in \mathbb{F}_{q}\setminus\{0,1\}$ such that $\theta_{x,i}-1$ is invertible,
$x=1,2,\ldots,r$, $\ell = 1,2,\ldots,s$ and $j = 1,2,\ldots,r$.
\end{enumerate}
We denote the above construction of our second piggybacking codes as
$\mathcal{C}_2(n,k,m,L)$.
In the construction of $\mathcal{C}_2(n,k,m,L)$, we design the piggyback functions
in step three to repair the single-node failure of data nodes and employ the transformation
functions designed in step four to repair the single-node failure of parity nodes.
The piggyback structure of $\mathcal{C}_2(n,k,m,L)$ is similar to that of our first piggybacking
codes $\mathcal{C}_1(n,k,m,L)$. The difference is that the piggyback functions of
$\mathcal{C}_2(n,k,m,L)$ are designed for data node repair, while the piggyback functions of
$\mathcal{C}_1(n,k,m,L)$ are designed for the repair of both data node and parity node.

Given that $x>i$, we have $f''_x(\boldsymbol{a}^{(\ell)}_i)=\theta_{x,i}\cdot f'_{x}(\boldsymbol{a}^{(\ell)}_i) + f'_i(\boldsymbol{a}^{(\ell)}_{x})$
and $f''_i(\boldsymbol{a}^{(\ell)}_x)=f'_{i}(\boldsymbol{a}^{(\ell)}_x) + f'_x(\boldsymbol{a}^{(\ell)}_{i})$.
It is easy to check that we can compute $f'_{i}(\boldsymbol{a}^{(\ell)}_x)$ and $f'_x(\boldsymbol{a}^{(\ell)}_{i})$
from $f''_x(\boldsymbol{a}^{(\ell)}_i)$ and $f''_i(\boldsymbol{a}^{(\ell)}_x)$.
We can also compute $f''_i(\boldsymbol{a}^{(\ell)}_x)$ from any two out of the three symbols
\[
f'_{i}(\boldsymbol{a}^{(\ell)}_x),f'_x(\boldsymbol{a}^{(\ell)}_{i}),f''_x(\boldsymbol{a}^{(\ell)}_i).
\]

\begin{example}\label{ex3}
Consider the example of $(n,k,m,L)=(12,8,16,2)$, the $k=8$ data nodes are divided into
two subsets $\Phi_1 = \{1,2,3,4\}$ and $\Phi_2 = \{5,6,7,8\}$, and the $m=16$
columns are divided into four subsets $W_{\ell} = \{4(\ell-1)+1,4(\ell-1)+2,4(\ell-1)+3,4(\ell-1)+4\}$
with $\ell=1,2,3,4$. According to Eq. \eqref{gx}, the piggyback functions are defined as follows.
\begin{eqnarray}\nonumber
g(i,1) &=& a^{(1)}_{i,1} + a^{(1)}_{i,2} + a^{(1)}_{i,3} + a^{(1)}_{i,4} \text{ for } i=1,2,3,4,\\ \nonumber
g(4+i,1) &=& a^{(2)}_{i,1} + a^{(2)}_{i,2} + a^{(2)}_{i,3} + a^{(2)}_{i,4} \text{ for } i=1,2,3,4,\\ \nonumber
g(8+i,1) &=& a^{(3)}_{i,1} + a^{(3)}_{i,2} + a^{(3)}_{i,3} + a^{(3)}_{i,4} \text{ for } i=1,2,3,4,\\ \nonumber
g(i,2) &=& a^{(1)}_{i,5} + a^{(2)}_{i,6} + a^{(1)}_{i,8} \text{ for } i=1,2,3,4,   \\ \nonumber
g(4+i,2) &=& a^{(2)}_{i,5} + a^{(1)}_{i,7} + a^{(2)}_{i,8} \text{ for } i=1,2,3,4, \\ \nonumber
g(8+i,2) &=& a^{(1)}_{i,6} + a^{(2)}_{i,7}  \text{ for } i=1,2,3,4. \nonumber
\end{eqnarray}
The transformation functions are given in Eq. \eqref{eq:trans} with $x=1,2,3,4$, $i=1,2,3,4$
and $\ell=1,2,3,4$, where the symbol $f'_x(\boldsymbol{a}_{i}^{(\ell)})$ is as follows,
\begin{eqnarray}\nonumber
\left( f'_1(\boldsymbol{a}_{1}^{(\ell)}),f'_2(\boldsymbol{a}_{2}^{(\ell)}) \right)=& \left(f_1(\boldsymbol{a}_{1}^{(\ell)}), f_1(\boldsymbol{a}_{2}^{(\ell)})\right), \\ \nonumber
\left(f'_3(\boldsymbol{a}_{3}^{(\ell)}),f'_4(\boldsymbol{a}_{4}^{(\ell)})\right) =& \left(f_1(\boldsymbol{a}_{3}^{(\ell)}),f_1(\boldsymbol{a}_{4}^{(\ell)})\right) , \\ \nonumber
\left(f'_1(\boldsymbol{a}_{2}^{(\ell)}),f'_1(\boldsymbol{a}_{3}^{(\ell)})\right) =& \left(f_4(\boldsymbol{a}_{2}^{(\ell)}) + g(1,5 - \ell),f_4(\boldsymbol{a}_{3}^{(\ell)}) + g(2, 5 - \ell)\right),\\ \nonumber
\left(f'_1(\boldsymbol{a}_{4}^{(\ell)}),f'_2(\boldsymbol{a}_{1}^{(\ell)})\right) =& \left(f_4(\boldsymbol{a}_{2}^{(\ell)}) + g(3, 5 - \ell), f_2(\boldsymbol{a}_{1}^{(\ell)}) + g(4, 5 - \ell)\right) , \\ \nonumber
\left(f'_2(\boldsymbol{a}_{3}^{(\ell)}),f'_2(\boldsymbol{a}_{4}^{(\ell)})\right) =& \left(f_4(\boldsymbol{a}_{3}^{(\ell)}) + g(5, 5 - \ell), f_3(\boldsymbol{a}_{4}^{(\ell)}) + g(6, 5 - \ell)\right), \\ \nonumber
\left(f'_3(\boldsymbol{a}_{1}^{(\ell)}),f'_3(\boldsymbol{a}_{2}^{(\ell)})\right) =& \left(f_3(\boldsymbol{a}_{1}^{(\ell)}) + g(7, 5 - \ell), f_2(\boldsymbol{a}_{2}^{(\ell)}) + g(8, 5 - \ell)\right), \\ \nonumber
\left(f'_3(\boldsymbol{a}_{4}^{(\ell)}),f'_4(\boldsymbol{a}_{1}^{(\ell)})\right) =& \left(f_4(\boldsymbol{a}_{4}^{(\ell)}) + g(9, 5 - \ell), f_4(\boldsymbol{a}_{1}^{(\ell)}) + g(10, 5 - \ell)\right) , \\ \nonumber
\left(f'_4(\boldsymbol{a}_{2}^{(\ell)}),f'_4(\boldsymbol{a}_{3}^{(\ell)})\right) =& \left(f_3(\boldsymbol{a}_{2}^{(\ell)}) + g(11, 5 - \ell),  f_2(\boldsymbol{a}_{3}^{(\ell)})+ g(12, 5 - \ell)\right) , \nonumber
\end{eqnarray}
where $\ell = 1,2,3,4$ and $g(\cdot,x) = 0 $ for $x\geq 3$.

The code $\mathcal{C}_2(12,8,16,2)$	is shown in Fig. \ref{sr}.
\end{example}

\begin{figure*}[htbp]
	\centering
	\includegraphics[width=1.0\textwidth]{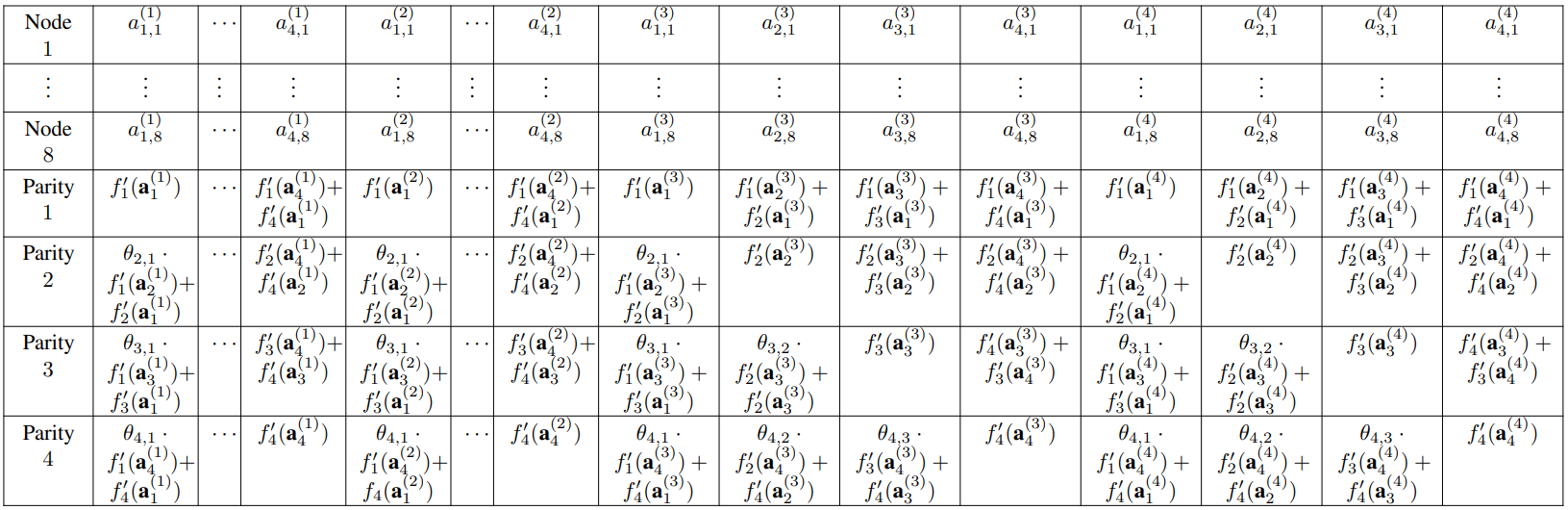}
	\caption{Code $\mathcal{C}_2(12, 8, 16, 2)$.}
	\label{sr}
\end{figure*}

\section{Repair Method of $\mathcal{C}_2(n,k,m,L)$}\label{sec:repair2}
In this section, we present repair method for any single-node erasure of
$\mathcal{C}_2(n,k,m,L)$.

\subsection{Repair Method for Data Node}
The repair method for data node of $\mathcal{C}_2(n,k,m,L)$ is quite similar to that
of $\mathcal{C}_1(n,k,m,L)$, because the two codes have the same structure of piggyback
function. The difference is that the piggyback function of $\mathcal{C}_1(n,k,m,L)$
is computed from both data symbols and parity symbols, while the piggyback function of
$\mathcal{C}_2(n,k,m,L)$ is computed from only data symbols.
Suppose that data node $t \in \Phi_u$ fails, where $t\in\{1,2,\ldots,k\}$ and
$u\in\{1,2,\ldots,L\}$, the repair method is as follows.
\begin{enumerate}
	\item Recall that there is one parity symbol in each column which is not added
by piggyback function. We can recover the last $ur$ erased symbols in node $t$ by
downloading $(k-1)ur$ data symbols in the last $ur$ columns in the first $k$ rows except
row $t$ and downloading the $ur$ parity symbols in the last $ur$ columns of
without attaching piggyback function.
	\item We download the parity symbols of which the corresponding piggyback functions
are computed from the symbols in node $t$, together with
$f_x(\boldsymbol{a}^{(\ell)}_i)$ with $x= 1,2,\ldots,r$, $i = 1,2,\ldots,r$ and
$\ell = s-u+1,\ldots,s$, to recover the first $(s-u)r$ symbols in node $t$.
\end{enumerate}

Continue the code in Example \ref{ex3}.
Suppose that node $t=1$ fails, we have $u=1$ and we can recover the last four symbols
$a^{(4)}_{i,1}$ with $i=1,2,3,4$ in node $1$ and the 16 parity symbols $f_{x}(\boldsymbol{a}^{(4)}_i)$
with $x = 1,2,3,4$ and $i = 1,2,3,4$, by downloading the 28 symbols $a^{(4)}_{i,j}$ with $j=2,3,\ldots,8$
and $i=1,2,3,4$, and the four parity symbols
$$f_1({\boldsymbol{a}}^{(4)}_1), f_1({\boldsymbol{a}}^{(4)}_2), f_1({\boldsymbol{a}}^{(4)}_3), f_1({\boldsymbol{a}}^{(4)}_4).$$
Then we download the following 12 parity symbols
\begin{align*}
& f'_1(\boldsymbol{a}^{(4)}_2)+f'_2(\boldsymbol{a}^{(4)}_1) , f'_1(\boldsymbol{a}^{(4)}_3)+f'_3(\boldsymbol{a}^{(4)}_1), \\
& f'_2(\boldsymbol{a}^{(4)}_3)+f'_3(\boldsymbol{a}^{(4)}_2) , f'_2(\boldsymbol{a}^{(4)}_4) + f'_4(\boldsymbol{a}^{(4)}_2),\\
&f'_4(\boldsymbol{a}^{(4)}_3) + f'_3(\boldsymbol{a}^{(4)}_4), f'_1(\boldsymbol{a}^{(4)}_4)+f'_4(\boldsymbol{a}^{(4)}_1) , \\
&\theta_{4,1}\cdot f'_1(\textbf{a}^{(4)}_4) + f'_4(\textbf{a}^{(4)}_1), \theta_{2,1}\cdot f'_1(\boldsymbol{a}^{(4)}_2) + f'_2(\boldsymbol{a}^{(4)}_1),\\
& \theta_{3,1}\cdot f'_1(\boldsymbol{a}^{(4)}_3) + f'_3(\boldsymbol{a}^{(4)}_1), \theta_{3,2}\cdot f'_2(\boldsymbol{a}^{(4)}_3) + f'_2(\boldsymbol{a}^{(4)}_3), \\
&\theta_{4,2}\cdot f'_2(\textbf{a}^{(4)}_4) + f'_4(\textbf{a}^{(4)}_2), \theta_{4,3}\cdot f'_3(\textbf{a}^{(4)}_4) + f'_4(\textbf{a}^{(4)}_3),
\end{align*}
and the $36$ data symbols $a^{(\ell)}_{i,j}$ with $\ell=1,2,3$,  $i=1,2,3,4$ and $j=2,3,4$
to recover the first 12 data symbols in node $1$. The repair bandwidth of node $1$ is $80$
symbols. Similarly, we can calculate that the repair bandwidth of each node in
$\{2,3,4\}$ is 80 symbols, the repair bandwidth of each node in $\{5,8\}$ is 90 symbols,
and the repair bandwidth of each node in $\{6,7\}$ is 86 symbols.

\subsection{Repair Process for Parity Node}\label{piggyback2parity}

Suppose that node $t$ fails, where $t\in\{k+1,k+2,\ldots,k+r\}$. The repair method of
node $t$ is as follow.
\begin{enumerate}
\item We compute the $sr$ symbols $f_{x}({\boldsymbol{a}}^{(\ell)}_{t-k})$ with
$\ell =1,2,\ldots,s$ and $x = 1,2,\ldots,r$ by downloading the $sk$ symbols in the
first $k$ rows in columns $(\ell-1)r+t-k$ with $\ell = 1,2,\ldots,s$.
\item We recover the erased $sr$ symbols $f''_{t-k}({\boldsymbol{a}}^{(\ell)}_{i})$
with $\ell = 1,2,\ldots,s$ and $i = 1,2,\ldots,r$ by downloading the data symbols which are
used in computing the piggyback functions located in the $r$ rows $k+1,k+2,\ldots,k+r$
in the $L$ columns $(\ell-1)r+t-k$ with $\ell=s-L+1,s-L+2,\ldots,s$ and downloading
the $s(r-1)$ symbols $f''_{x}({\boldsymbol{a}}^{(\ell)}_{t-k})$ with $\ell =1,2,\ldots,s$
and $x = 1,2,\ldots,t-k-1,t-k+1,\ldots,r$.
\end{enumerate}

Continue the code in Example \ref{ex3}. Suppose that node $t=k+1$ fails, we download
$32$ data symbols $a^{(\ell)}_{1,j}$ with $\ell=1,2,3,4$ and $j=1,2,\ldots,8$
to calculate $f_x({\boldsymbol{a}}^{(\ell)}_1)$ with $x = 1,2,3,4$ and $\ell = 1,2,3,4$.
Then we download the following 20 data symbols
\begin{align*}
	& a^{(1)}_{4,1}, a^{(1)}_{4,2}, a^{(1)}_{4,3}, a^{(1)}_{4,4}, a^{(2)}_{3,1}, a^{(2)}_{3,2}, a^{(2)}_{3,3}, a^{(2)}_{3,4}, \\
	& a^{(3)}_{2,1}, a^{(3)}_{2,2}, a^{(3)}_{2,3}, a^{(3)}_{2,4}, a^{(1)}_{4,5}, a^{(2)}_{4,6}, a^{(1)}_{4,8}, a^{(2)}_{3,5}, \\
	& a^{(1)}_{3,7}, a^{(2)}_{3,8}, a^{(1)}_{2,6}, a^{(2)}_{2,7},
\end{align*}
to calculate the piggyback functions
\[
g(4,1),g(7,1),g(10,1),g(4,2),g(7,2),g(10,2).
\]
Together with $f_x({\boldsymbol{a}}^{(\ell)}_1)$ for $x = 1,2,3,4$ and $\ell = 1,2,3,4$,
and the above 6 piggyback functions, we can recover all the symbols in node $t$ by
downloading the following symbols.
\begin{align*}
& f'_1(\textbf{a}^{(3)}_2)+f'_2(\textbf{a}^{(3)}_1), f'_1(\textbf{a}^{(3)}_3)+f'_3(\textbf{a}^{(3)}_1), \\
&f'_1(\textbf{a}^{(3)}_4)+f'_4(\textbf{a}^{(3)}_1), f'_1(\textbf{a}^{(4)}_2)+f'_2(\textbf{a}^{(4)}_1), \\
&f'_1(\textbf{a}^{(4)}_3)+f'_3(\textbf{a}^{(4)}_1), f'_1(\textbf{a}^{(4)}_4)+f'_4(\textbf{a}^{(4)}_1).
\end{align*}

\subsection{Repair Bandwidth of $\mathcal{C}_2(n,k,m,L)$}
When $L$ is a factor of $k$, with similar proof in Lemma \ref{lemma3}, we can show the
bounds of repair bandwidth of data nodes of $\mathcal{C}_2(n,k,m,L)$ in the next lemma
and we omit the proof.

\begin{lemma}\label{lemma6}
If $L$ is a factor of $k$, the lower bound $\gamma_{min}^{sys}$ and the upper bound
$\gamma_{max}^{sys}$ of the average repair ratio of data nodes of $\mathcal{C}_2(n,k,m,L)$ is
	\begin{eqnarray}\nonumber
		&\gamma_{min}^{sys} = \frac{L+1}{2s} + \frac{(s^2-s(L+1)+\frac{(L+1)(2L+1)}{6})}{Ls(r-1)} +\frac{(r-1)(L-3)}{2ksr}, \\ \nonumber
		&\gamma_{max}^{sys} = \gamma_{min}^{all} + \frac{L(r-1)}{4sk^2}. \nonumber
	\end{eqnarray}
\end{lemma}

By Lemma \ref{lemma6}, we have $|\gamma_{sys} - \gamma_{min}^{sys}| \leq
|\gamma_{max}^{sys} - \gamma_{min}^{sys}| = \frac{L(r-1)}{4sk^2}$. When $k \to \infty$,
we have $\gamma^{sys} = \gamma^{sys}_{min}$. Therefore, we do not distinguish
between $\gamma^{sys}$ and $\gamma^{sys}_{min}$ in the rest of the paper.

Similar to the proof of Lemma \ref{lemma4}, we can also show that the minimum value of average repair bandwidth ratio $\gamma^{sys}$
is achieved when $L = \sqrt{\frac{6s^2-6s+1}{3r-1}}$.

\begin{lemma}
If $L$ is a factor of $k$, the average repair bandwidth ratio of parity nodes for codes $\mathcal{C}_2(n,k,m,L)$ is
	\begin{equation}
		\gamma_{parity} = \frac{2}{r}+\frac{1}{k} - \frac{1}{kr} -\frac{L+1}{2sr}.
	\end{equation}
\end{lemma}

\begin{proof}
According to the repair methods in Section \ref{piggyback2parity}, we need to
download $\sum_{i=1}^{r}sk$ symbols in repairing each of the $r$ parity nodes in the first step.
In the second step, we need to download $\sum_{i=1}^{r}s(r-1) + 2\sum_{i = 1}^{L} |\Phi_i|r(s-i)$
in repairing each of the $r$ parity nodes. Therefore, we can calculate that
	\begin{equation}\nonumber
		\gamma^{parity} = \frac{\sum_{i = 1}^{r}s(k+r-1) + \sum_{i=1}^{L}|\Phi_i|r(s-i)}{skr^2}.
	\end{equation}
	Because $L$ is a factor of $k$, we have $|\Phi_i| = \frac{k}{L}$ for $i = 1,2,\ldots,L$,
and further obtain that
$$\gamma_{parity} = \frac{2}{r}+\frac{1}{k} - \frac{1}{kr} -\frac{L+1}{2sr}.$$
\end{proof}

\section{Comparison}\label{sec:comparison}
In this section, we evaluate the average repair bandwidth of all $n$ nodes for our two codes and the
existing piggybacking codes with low repair bandwidth.

\subsection{Piggybacking Codes with $m < r$}
\label{sec:com-1}
\begin{figure*}[htbp]
	\centering
	\begin{minipage}{0.49\linewidth}
		\centering
		\includegraphics[width=1.0\linewidth]{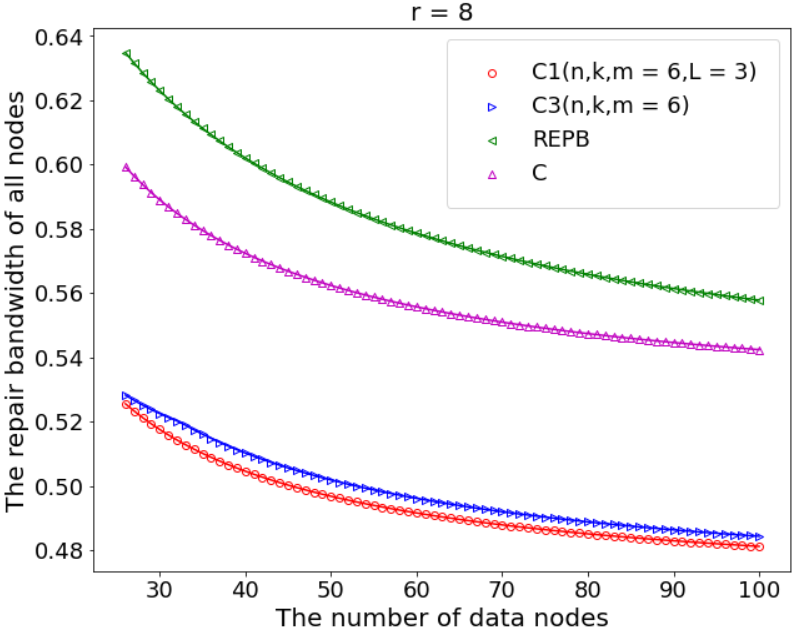}
	\end{minipage}
	\begin{minipage}{0.49\linewidth}
		\centering
		\includegraphics[width=1.0\linewidth]{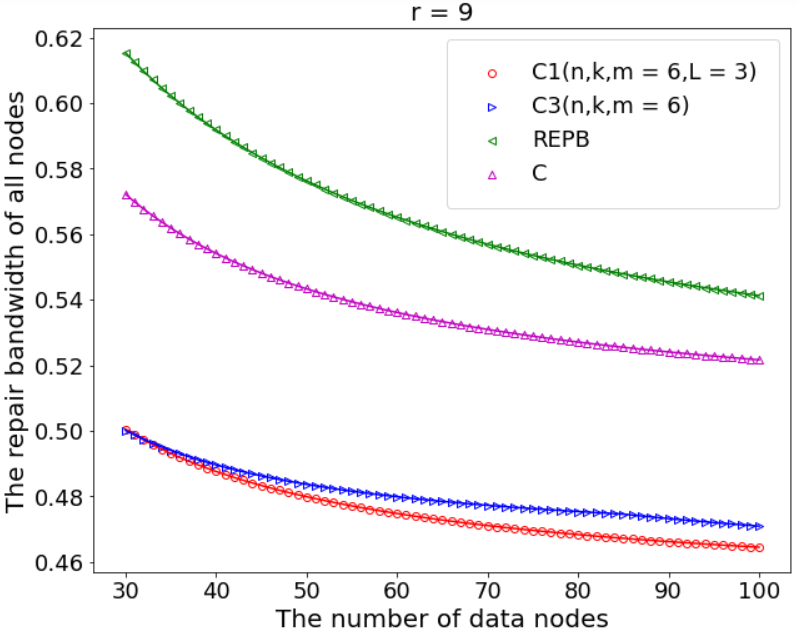}
	\end{minipage}
	\caption{The average repair bandwidth ratio of all nodes for the proposed
		codes $\mathcal{C}_1$,  REPB codes \cite{2018Repair}, codes $\mathcal{C}_3$ \cite{2022arXiv220509659W} and codes $\mathcal{C}$ \cite{2022arXiv220514555J}, where $r = 8,9$, $m = 6$ and $k = 30,31,\ldots,100$.}
	\label{fig:11}
\end{figure*}

First, we evaluate $\mathcal{C}_1(n,k,m,L)$ and the existing piggybacking codes
\cite{2018Repair,2022arXiv220509659W,2022arXiv220514555J}
such that the sub-packetization is no larger than $r$.
Denote the codes in \cite{2022arXiv220509659W} as $\mathcal{C}_3(n, k, m)$,
the MDS codes (the first codes) in \cite{2022arXiv220514555J} as $\mathcal{C}$,
the codes in \cite{2018Repair} as REPB.

Fig. \ref{fig:11} shows the evaluations for $r = 8,9$, $m = 6$ and $k = 30,31,\ldots,100$.
Note that the codes in our conference paper \cite{Shi2206:Piggybacking} can only support
the parameter $m=r$ and do not draw the points for codes \cite{Shi2206:Piggybacking} in Fig. \ref{fig:11}.
The results in Fig. \ref{fig:11} demonstrate that our codes $\mathcal{C}_1(n,k,m,L)$
have lower repair bandwidth than all existing piggybacking codes when $m < r$ for all the evaluated
parameters. 

\subsection{Piggybacking Codes with $m\geq r$}
\begin{figure*}[htbp]
	\centering
	\begin{minipage}{0.3\linewidth}
		\centering
		\includegraphics[width=1.0\linewidth]{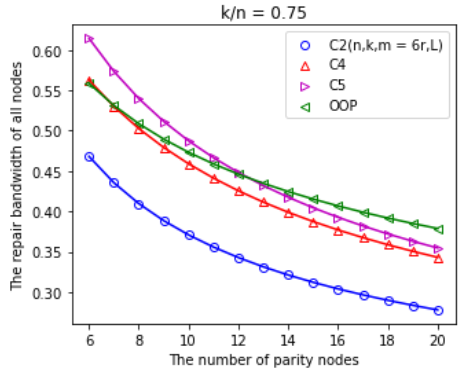}
	\end{minipage}
	\begin{minipage}{0.3\linewidth}
		\centering
		\includegraphics[width=1.0\linewidth]{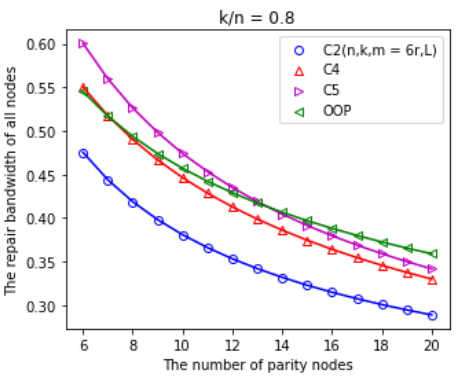}
	\end{minipage}
	\begin{minipage}{0.3\linewidth}
		\centering
		\includegraphics[width=1.0\linewidth]{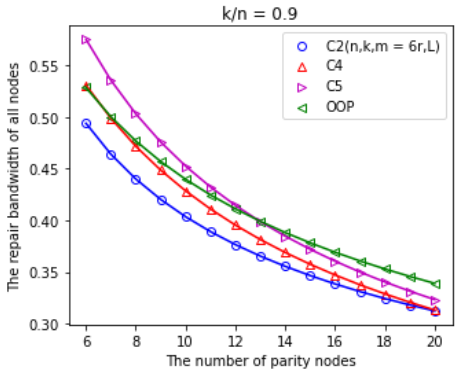}
	\end{minipage}
	
	\begin{minipage}{0.3\linewidth}
		\centering
		\includegraphics[width=1.0\linewidth]{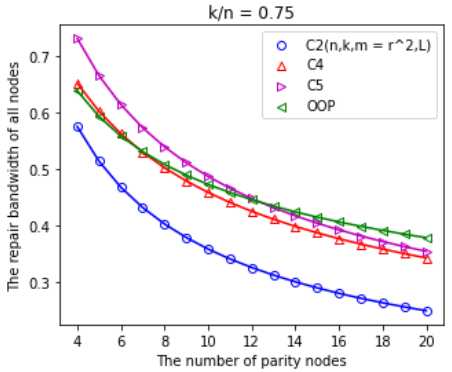}
	\end{minipage}
	\begin{minipage}{0.3\linewidth}
		\centering
		\includegraphics[width=1.0\linewidth]{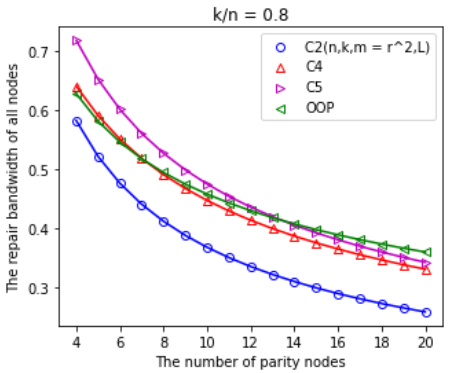}
	\end{minipage}
	\begin{minipage}{0.3\linewidth}
		\centering
		\includegraphics[width=1.0\linewidth]{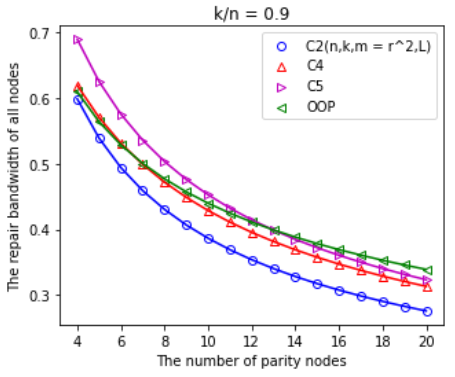}
	\end{minipage}
	\caption{The average repair bandwidth ratio of all nodes for the proposed
		codes $\mathcal{C}_2$,  OOP codes \cite{2019AnEfficient}, codes $\mathcal{C}_4$ \cite{Shi2206:Piggybacking} and codes $\mathcal{C}_5$ \cite{2021piggybacking}, where $r=4,5,6,7,\ldots,20$, code rate $\frac{k}{n}=0.75,0.8,0.9$ and $m = 6r, r^2$.}
	\label{fig:111}
\end{figure*}

In the following, we evaluate $\mathcal{C}_2(n,k,m,L)$ and the existing piggybacking codes
with $m\geq r$.


We label the codes in \cite{Shi2206:Piggybacking} as $\mathcal{C}_4$, the
codes in \cite{2021piggybacking} as $\mathcal{C}_5$.
Fig. \ref{fig:111} shows the average repair bandwidth ratio of all nodes for
our codes $\mathcal{C}_2(n,k,m,L)$ and the existing piggybacking codes, including
OOP, $\mathcal{C}_4$ and $\mathcal{C}_5$ when $r=4,5,\ldots,20$,
$\frac{k}{n}=0.75,0.8,0.9$ and $m = 6r, r^2$. The results show that our
codes $\mathcal{C}_2(n,k,m,L)$ have lower repair bandwidth than the other codes
for all the evaluated parameters. The essential reason of lower repair bandwidth of
our $\mathcal{C}_2(n,k,m,L)$ is that we jointly design the piggyback function for data node
repair and the transformation function for parity node repair.


\section{Conclusion}\label{sec:con}

In this paper, we design two classes of piggybacking codes with flexible sub-packetization level. The first piggybacking codes $\mathcal{C}_1$ have the lower repair bandwidth for single-node failure than all exsiting piggybacking codes with $m\leq r$ for all the evaluated parameters. Our second piggybacking codes can support sub-packetization $m=sr$ with $2\leq s\leq r$ that have the lowest average repair bandwidth for all nodes among all existing piggybacking codes for the evaluated parameters. The piggybacking
codes constructions by jointly designing piggyback functions and transformation functions
to support more larger sub-packetization, say $m=sr$ with $s>r$, to further reduce repair bandwidth
is one of our future work.

\ifCLASSOPTIONcaptionsoff
\newpage
\fi

\bibliographystyle{IEEEtran}
\bibliography{ck}
\end{document}